\documentclass[11pt]{scrartcl}

\usepackage{amsmath}
\usepackage{amssymb}
\usepackage{amsthm}
\usepackage{framed}

\newtheorem{theorem}{Theorem}

\newtheorem{corollary}{Corollary}

\newtheorem{definition}{Definition}

\newcommand{\citep}{\cite}

\title{Truthful Fair Division Without\\ Free Disposal}
\author{
Xiaohui Bei\\Nanyang Technological University
\and
Guangda Huzhang\\Nanyang Technological University
\and
Warut Suksompong\\University of Oxford
}
\date{\vspace{-3ex}}

\begin{document}

\maketitle

\begin{abstract}
We study the problem of fairly dividing a heterogeneous resource, commonly known as cake cutting and chore division, in the presence of strategic agents. While a number of results in this setting have been established in previous works, they rely crucially on the \emph{free disposal} assumption, meaning that the mechanism is allowed to throw away part of the resource at no cost. In the present work, we remove this assumption and focus on mechanisms that always allocate the entire resource. We exhibit a truthful and envy-free mechanism for cake cutting and chore division for two agents with piecewise uniform valuations, and we complement our result by showing that such a mechanism does not exist when certain additional constraints are imposed on the mechanisms. Moreover, we provide bounds on the efficiency of mechanisms satisfying various properties, and give truthful mechanisms for multiple agents with restricted classes of valuations.
\end{abstract}

\section{Introduction}

Given a heterogeneous divisible resource and a set of interested agents with potentially differing valuations on different parts of the resource, how can we allocate the resource to the agents in such a way that all agents perceive the resulting allocation as fair? The resource is often modeled as a cake in the literature, and the problem, which therefore commonly goes by the name of \emph{cake cutting}, has occupied the minds of mathematicians, computer scientists, economists, and political scientists alike for the past seventy years \cite{brams1996fair,moulin2004fair,Pro16,RW98,steinhaus48}.
Cake in the cake cutting problem is used to represent a \emph{desirable} resource; all agents wish to maximize the amount of resource that they receive. In contrast, the dual problem to cake cutting, known as \emph{chore division}, aims to allocate an \emph{undesirable} resource to the agents, with every agent wanting to receive as little of the resource as possible. Though several algorithms for cake cutting also apply to chore division, the theoretical properties of the two problems differ in many cases, and much less work has been done on chore division than on cake cutting~\cite{dehghani2018chore,farhadi2017complexity,heydrich2015dividing,peterson1998exact,peterson2002four}.

Perhaps the simplest and most well-known fair division protocol is the \emph{cut-and-choose protocol}, which works for both cake cutting and chore division with two agents. The protocol operates by letting the first agent divide the resource into two parts that she values equally, and letting the second agent choose the part that she prefers. The resulting allocation is always \emph{envy-free}---each agent likes her part at least as much as the other agent's part, and \emph{proportional}---both agents find their part better than or equal to half of the entire resource. However, the protocol has the disadvantage that it is not \emph{truthful}, meaning that a strategic agent can sometimes benefit from misreporting her valuation to the protocol. For example, if the first agent values the whole cake equally, according to the protocol she will divide the cake into half and get half of her value for the entire cake. However, if she knows that the second agent only cares about the leftmost quarter of the cake, she can divide the cake into the leftmost quarter and the rest, knowing that the second agent will choose the left part and leave her with three-quarters of the cake. The failure to satisfy truthfulness renders the protocol difficult to participate in, since the first agent needs to guess the second agent's valuation in order to find a beneficial manipulation.

This issue was first addressed by Chen et al. \cite{CLPP13}, who gave a truthful deterministic cake cutting mechanism that is Pareto optimal, envy-free, and proportional for any number of agents with piecewise uniform valuations.
Chen et al.'s result shows that fairness and truthfulness are compatible in the allocation of heterogeneous resources. Nevertheless, their result hinges upon a pivotal assumption known as the \emph{free disposal} assumption, which says that the mechanism is allowed to throw away part of the resource without incurring any cost.\footnote{Note that free disposal does not preclude Pareto optimality.
The mechanism can throw away parts of the resource not valued by any agent and still maintain Pareto optimality; this is exactly what Chen et al.'s mechanism does.} While certain resources such as cake or machine processing time may be easy to get rid of, for other resources this is not the case. For instance, when we divide a piece of land among antagonistic agents or countries, we cannot simply throw away part of the land, and any piece of land left unallocated constitutes a potential subject of future dispute. The free disposal assumption is even less reasonable when it comes to chore allocation---indeed, with this assumption we might as well simply dispose of the entire chore altogether!

With this motivation in mind, we consider in the present paper the problem of fairly and truthfully dividing heterogeneous resources \emph{without the free disposal assumption}. Not having the ability to throw away part of the resource makes the task of the mechanism more complicated. 
For example, with free disposal allowed, the mechanism can throw away parts that are not valued by any agent, thereby preventing an agent from gaining by not reporting parts of the resource for which she is the only one who has positive value, in the hope of getting some of those parts for free along with a larger share of the remaining parts. As Chen et al. \cite{CLPP13} noted, getting rid of the free disposal assumption adds ``significant complexity'' to the problem, since the mechanism would have to specify exactly how to allocate parts that no agent desires. The same group of authors also gave an example illustrating that removing the assumption can be problematic even in the special case of two agents with very simple valuations. Indeed, could it be that there is an impossibility result once we dispose of free disposal?

\subsection{Our Results}

Throughout the paper, we focus on deterministic mechanisms that are required to allocate the entire resource, which we model as the interval $[0,1]$. We assume that agents have \emph{piecewise uniform} valuations, meaning that for each agent the cake can be partitioned into desired and undesired intervals, and the agent has the same marginal utility for any fractional piece of any desired interval. We investigate the compatibility of truthfulness and fairness in this setting.

First, in Section~\ref{sec:mechanism-two}, we show that truthfulness and fairness are compatible when there are two agents by exhibiting a truthful, envy-free, and Pareto optimal cake cutting mechanism (Theorem~\ref{thm:twoagents-cake}). At a high level, the mechanism lets the two agents ``eat'' their desired intervals of the cake at the same speed but starting from different ends of the cake.  Using a simple reduction from chore division to cake cutting, we also derive a chore division mechanism for two agents with the same set of properties (Theorem~\ref{thm:twoagents-chore}).

In Section~\ref{sec:impossibility}, we show that if we add certain requirements for the mechanism on top of being fair and truthful, then no desirable mechanism exists even for two agents. In particular, the impossibility holds when we make any one of the following assumptions in addition to truthfulness and envy-freeness: (i) \emph{anonymity}---the mechanism must treat both agents equally (Theorem~\ref{thm:anonymous}); (ii) \emph{connected piece assumption}---the mechanism must allocate a single interval to each agent (Theorem~\ref{thm:connected}); and (iii) \emph{position obliviousness}---the values that the agents receive depend only on the lengths of the pieces desired by various subsets of agents and not on the positions of these pieces (Theorem~\ref{thm:position-oblivious}). In fact, our first impossibility result still holds even if we remove envy-freeness, while for our second and third results we can replace envy-freeness by the significantly weaker fairness notion of \emph{positive share}, which requires that each agent receive a positive utility (for cake cutting) or not incur the entire cost (for chore division). Our results in Sections~\ref{sec:mechanism-two} and \ref{sec:impossibility} are summarized in Table~\ref{table:twoagents}.

\begin{table}
\centering
    \begin{tabular}{| c | c |}
    \hline
     \textbf{Properties of mechanism} & \textbf{Existence} \\ \hline
     Truthful + envy-free + Pareto optimal & Yes (Theorems~\ref{thm:twoagents-cake} and \ref{thm:twoagents-chore}) \\ \hline
     Truthful + connected + position oblivious & Yes (trivial mechanism)  \\ \hline
     Truthful + anonymous & No (Theorem~\ref{thm:anonymous} and Corollary~\ref{cor:chore-impossibility})   \\ \hline
     Truthful + connected + positive share & No (Theorem~\ref{thm:connected} and Corollary~\ref{cor:chore-impossibility}) \\ \hline
     Truthful + position oblivious + positive share & No (Theorem~\ref{thm:position-oblivious} and Corollary~\ref{cor:chore-impossibility}) \\
    \hline
    \end{tabular}
    \vspace{5mm}
    \caption{Existence of mechanisms satisfying truthfulness in combination with other properties in the case of two agents. By ``trivial mechanism'' we refer to a mechanism that allocates the entire resource to one fixed agent. All results hold for both cake cutting and chore division.}
    \label{table:twoagents}
\end{table}

In Section~\ref{sec:competitive}, we investigate the efficiency of mechanisms satisfying combinations of truthfulness, envy-freeness, and Pareto optimality. Our metric is the \emph{efficiency ratio}, a standard measure for the welfare performance of mechanisms.\footnote{The efficiency ratio is closely related to the \emph{price of fairness}, which measures the worst-case welfare loss due to imposing fairness constraints \cite{BeiLuMa19,CaragiannisKaKa11,heydrich2015dividing}.} For cake cutting, we show that our mechanism in Section~\ref{sec:mechanism-two} achieves an efficiency ratio of $3/4$, and that this is optimal among all truthful and envy-free mechanisms (Theorem~\ref{thm:cake-TR-EF} and Corollary~\ref{cor:cake-eating}). By contrast, for chore division we show that our mechanism achieves the worst possible efficiency ratio of 0; however, this is necessarily the case for any truthful and envy-free mechanism (Theorem~\ref{thm:chore-TR-EF}). Moreover, we provide tight bounds on the efficiency ratio of mechanisms satisfying each of envy-freeness and Pareto optimality, for both cake cutting and chore division (Theorems~\ref{thm:cake-EF}, \ref{thm:cake-PO}, and \ref{thm:chore-EF}).

Finally, in Section~\ref{sec:mechanism-many}, we consider the more general setting where there are multiple agents. We assume that each agent only values a single interval of the form $[0,x_i]$. We present a truthful, envy-free, and Pareto optimal cake cutting mechanism (Theorem~\ref{thm:many-cake}) and a truthful, proportional, and Pareto optimal chore division mechanism (Theorem~\ref{thm:many-chore}) for any number of agents with valuations in this class.

\subsection{Related Work}

Cake cutting has been a central topic in the area of social choice and economics for decades.
While the existence and computation of fair allocations have been extensively studied~\cite{aziz2016discrete,aziz2015discrete,BT95,dubins1961cut,GoldbergHoSu20,Stro80,su1999}, the work of Chen et al. \cite{CLPP13} that we mentioned earlier was the first to consider incentive issues.
As with Chen et al., Maya and Nisan \cite{MN12} considered piecewise uniform valuations and gave a characterization of truthful and Pareto optimal mechanisms for two agents. Recently, Alijani et al. \cite{AFGST17} presented a truthful and envy-free mechanism in the setting where every agent only values a single interval.

For valuation functions beyond piecewise uniform, most results are negative.
For example, for piecewise constant valuations, Aziz and Ye \cite{aziz2014cake} showed that there is no truthful and robust proportional mechanism,\footnote{Aziz and Ye \cite{aziz2014cake} call an allocation \emph{robust proportional} if it satisfies the following property: even if an agent perturbs her value density function, as long as the ordinal information of the function is unchanged, then the allocation remains proportional. Robust proportionality is a stronger notion than standard proportionality.} two works~\cite{BCHTW17,Menon2017} showed that there is no proportional mechanism that allocates connected pieces or is non-wasteful, and Bei et al. \cite{BCHTW17} also showed that there is no truthful mechanism that satisfies position obliviousness.
In the Robertson-Webb query model, Kurokawa et al. \cite{kurokawa2013cut} showed that there is no truthful and envy-free mechanism with bounded queries, while Br\^{a}nzei and Miltersen \cite{branzei2015dictatorship} proved that any deterministic truthful mechanism for two agents must be a dictatorship.

In all of the works above, either the free disposal assumption is made, or it is assumed that every piece of the cake is valuable for at least one agent. In contrast, in our work the mechanism is required to always allocate the entire cake.\footnote{On the other hand, in the study of truthful mechanisms for allocating \emph{indivisible} items, it is commonly assumed that the mechanism is required to allocate all items~\cite{AmanatidisBiCh17,AmanatidisBiMa16}.}
Furthermore, all aforementioned results are restricted to deterministic mechanisms. If one allows randomization, several truthful-in-expectation mechanisms that guarantee either proportionality or envy-freeness have been proposed~\cite{branzei2015dictatorship,CLPP13,MT10}.

Finally, preferences analogous to piecewise uniform valuations have been considered in other settings including matching~\cite{Bade15,BogomolnaiaMo04} and collective choice~\cite{AzizBoMo17,BogomolnaiaMoSt05,Duddy15}, where they are known as \emph{dichotomous preferences}. Fairness notions corresponding to proportionality and positive share have also been studied in the collective choice setting.

\section{Preliminaries}
\label{sec:prelim}

We consider a heterogeneous divisible resource, which we represent by the interval $[0,1]$. A \emph{piece} of the resource is a finite union of disjoint intervals. The resource is to be allocated to $n$ agents $a_1,a_2,\dots,a_n$. Each agent $a_i$ has a \emph{density function} $f_i:[0,1]\rightarrow\mathbb{R}^+\cup\{0\}$, which captures how the agent values different parts of the resource.
We assume that the agents have \emph{piecewise uniform} valuations, i.e., for each agent $a_i$, the density function $f_i$ takes on the value $1$ on a finite set of intervals and $0$ on the remaining intervals. The value of agent $a_i$ for a subset $S\subseteq[0,1]$ is defined as $v_i(S)=\int_S f_i \ dx$, which is equivalent to the total length of the intervals in $S$ on which $f_i$ takes on the value $1$.\footnote{In some papers, valuations are normalized so that $v_i([0,1])=1$ for all $i$. We do not follow this convention except in Section~\ref{sec:competitive}.} For chore division we will refer to the value as \emph{cost}. Let $W_i\subseteq [0,1]$ denote the piece on which $f_i=1$. We refer to a setting with agents and their density functions as an \emph{instance}.

An \emph{allocation} of the resource is denoted by a vector $A=(A_1,A_2,\dots,A_n)$, where $A_i$ is a union of finitely many intervals\footnote{We assume without loss of generality that these intervals are closed intervals. The assumption that each $A_i$ is a union of finitely many intervals is standard in the cake-cutting literature (see, e.g., \cite{Pro16}).} that represents the piece of the resource allocated to $a_i$, and $A_i\cap A_j$ has measure zero for any $i\neq j$. We consider two different types of resources: desirable resources, which we represent by a \emph{cake}, and undesirable resources, which we represent by a \emph{chore}. We refer to the problem of allocating the two types of resources as \emph{cake cutting} and \emph{chore division} respectively. The agents want to maximize their value for their allocated piece in cake cutting and minimize their cost for their allocated piece in chore division.

We are now ready to define the fairness properties that we consider in this paper.

\begin{definition}
In cake cutting, we say that an allocation $(A_1,A_2,\dots,A_n)$ satisfies
\begin{itemize}
\item \emph{envy-freeness}, if for every agent $a_i$, we have $v_i(A_i)\geq v_i(A_j)$ for any $j$;
\item \emph{proportionality}, if for every agent $a_i$, we have $v_i(A_i)\geq v_i([0,1])/n$;
\item \emph{positive share}, if for every agent $a_i$, we have $v_i(A_i)>0$.
\end{itemize}
In chore division, envy-freeness and proportionality are defined analogously but with the inequality signs reversed, while positive share is defined by the inequality $v_i(A_i)<v_i([0,1])$.
\end{definition}

A \emph{mechanism} is a function $\mathcal{M}:(f_1,f_2,\dots,f_n)\rightarrow (A_1,A_2,\dots,A_n)$ which, given the input density functions of the agents, computes an allocation for them. We only consider deterministic mechanisms in this paper, meaning that the allocation is completely determined by the input density functions. Moreover, we assume that the mechanism has to allocate the entire resource to the agents, i.e., in any allocation $(A_1,A_2,\dots,A_n)$ returned by the mechanism, $\bigcup_{i=1}^n A_i=[0,1]$. In other words, the mechanism does not have \emph{free disposal}. Note that when the entire resource is allocated, envy-free implies proportionality, and both notions are equivalent in the case of two agents. Moreover, both envy-freeness and proportionality are stronger than positive share.

We end this section by defining a number of properties of mechanisms that we consider in this paper. Given a vector of input density functions $\textbf{f}=(f_1,f_2,\dots,f_n)$, let $L_{\textbf{f}}$ be the indicator function that maps $\textbf{f}$ to a vector with $2^n$ components, where each component corresponds to a distinct subset of agents and the value of the component is the length of the piece desired only by that subset of agents (and not by any agent outside the subset).

\begin{definition}
A mechanism $\mathcal{M}:(f_1,f_2,\dots,f_n)\rightarrow (A_1,A_2,\dots,A_n)$ is said to satisfy
\begin{itemize}
\item \emph{envy-freeness}, if it always returns an envy-free allocation;
\item \emph{proportionality}, if it always returns a proportional allocation;
\item \emph{positive share}, if it always returns an allocation satisfying positive share;
\item \emph{truthfulness}, if it is a dominant strategy for every agent to report her true density function;
\item \emph{Pareto optimality}, if for any allocation returned by the mechanism, there does not exist another allocation that makes no agent worse off and at least one agent better off with respect to the same density functions;
\item the \emph{connected piece assumption}, if each $A_i$ is always a single interval;
\item \emph{anonymity}, if the following holds: Let $f_1,f_2,\dots,f_n$ be any density functions, and $\sigma$ be any permutation of $(1,2,\dots,n)$. If
$\mathcal{M}(f_1,f_2,\dots,f_n)=(A_1,A_2,\dots,A_n)$ and
$\mathcal{M}(f_{\sigma(1)},f_{\sigma(2)},\dots,f_{\sigma(n)})=(A'_1,A'_2,\dots,A'_n)$,
then $v_i(A_i)=v_i(A'_{\sigma^{-1}(i)})$ for every $i$.
\item \emph{position obliviousness}, if the following holds: Let $\emph{\textbf{f}}$ and $\emph{\textbf{f}}'$ be any vectors of density functions such that $L_{\emph{\textbf{f}}}=L_{\emph{\textbf{f}}'}$. If $\mathcal{M}(\emph{\textbf{f}})=(A_1,A_2,\dots,A_n)$ and $\mathcal{M}(\emph{\textbf{f}}')=(A'_1,A'_2,\dots,A'_n)$, then $v_i(A_i)=v_i'(A'_i)$ for every $i$.\footnote{This is a weaker notion of position obliviousness than the one considered by Bei et al.~\cite{BCHTW17}: Our definition only requires that the agents get the same value if the indicator function of their density functions remain the same, whereas Bei et al.'s definition also requires the pieces to be allocated in ``equivalent'' ways.}
\end{itemize}
\end{definition}

Intuitively, a mechanism is anonymous if the utility that the agents receive do not depend on the identities of the agents, and position oblivious if the values that the agents receive depend only on the lengths of the pieces desired by various subsets of agents and not on the positions of these pieces.

\section{Truthful Mechanisms for Two Agents}
\label{sec:mechanism-two}

In this section, we focus on the case of two agents. We show that in this case, there exists a truthful, envy-free, and Pareto optimal mechanism for both cake cutting and chore division, for two agents with arbitrary piecewise uniform valuations.

We first describe the cake cutting mechanism.

\begin{framed}
\noindent
\textbf{Mechanism~1} (for cake cutting between two agents) \\

\noindent
\emph{Step~1:} Find the smallest value of $x\in[0,1]$ such that $v_1([0,x])=v_2([x,1])$.\footnotemark \\

\noindent
\emph{Step~2:} Assign to $a_1$ the intervals in $[0,x]$ valued by $a_1$ and the intervals in $[x,1]$ \emph{not} valued by $a_2$, and assign the rest of the cake to $a_2$.
\end{framed}
\footnotetext{The existence of $x$ is guaranteed by the intermediate value theorem.}

While this is a succinct description of the mechanism, it turns out that the description is somewhat difficult to work with. We next provide an alternative formulation that is more intuitive and will help us in establishing the claimed properties of the mechanism.

\begin{framed}
\noindent
\textbf{Mechanism~1} (alternative formulation) \\

\noindent
\emph{Phase~1:} Let $a_1$ start at point $0$ of the cake moving to the right and $a_2$ start at point $1$ of the cake moving to the left. Let both agents ``eat'' the cake with the same constant speed, jumping over any interval for which they have no value according to their reported valuations.  If the agents are at the same point while both are still eating, go to Phase~3. Else, one of the agents has no more valued interval to eat; go to Phase~2. \\

\noindent
\emph{Phase~2:} Assume that $a_i$ is the agent who has no more valued interval to eat. Let $a_i$ stop and $a_{3-i}$ continue eating. If the agents are at the same point (either while $a_{3-i}$ eats or while $a_{3-i}$ jumps over an interval of zero value), go to Phase~3. Else, both agents have stopped but there is still unallocated cake between their current points. In this case, let $a_{3-i}$ continue eating the unallocated cake until she is at the same point as $a_i$, and go to Phase~3. \\

\noindent
\emph{Phase~3:} Assume that both agents are at point $x$ of the cake. (It is possible that the two agents meet while both of them are jumping. In this case, we let $a_2$ jump first.) Assign any unallocated interval to the left of $x$ to $a_2$ and any unallocated interval to the right of $x$ to $a_1$.
\end{framed}

From the perspective of this formulation, Mechanism~1 bears a resemblance to the \emph{probabilistic serial} mechanism for the random assignment problem~\cite{BogomolnaiaMo01}; both mechanisms let the agents simultaneously eat their desired part of the resource at the same speed.

We now prove the claimed properties of Mechanism~1.

\begin{theorem}
\label{thm:twoagents-cake}
Mechanism~1 is a truthful, envy-free, and Pareto optimal cake cutting mechanism for two agents.
\end{theorem}

\begin{proof}
We begin with truthfulness. Note that there is no incentive for an agent to report an interval that she does not value, since this can only result in the agent wasting time eating such intervals. So the only potential deviation is for the agent to report a strict subset of the intervals that she values. If the agent does not report intervals that she values, then the intervals that she jumps over before the agents meet will be lost to the other agent, and the agent can use the extra time gained from not reporting these intervals to eat intervals of no more than the same length. Moreover, not reporting intervals after the agents meet has no effect on the outcome of the mechanism.

Next, for envy-freeness, it suffices to show that each agent gets at least half of her valued intervals allocated in each phase. In Phase~1, each agent only gains intervals that she values, and loses intervals that she values (due to the other agent's eating) at no more than the same speed. In Phase~2, the agent who continues eating can only gain more, while the agent who has stopped eating has no more interval that she values. In Phase~3, $a_1$ has no unallocated interval to the left of $x$ that she values, so she cannot lose any unallocated interval that she values. The same argument holds for $a_2$.

Finally, our mechanism allocates any interval valued by at least one agent to an agent who values it. This establishes Pareto optimality.
\end{proof}

Mechanism~1 gives rise to a dual mechanism for two-agent chore division that satisfies the same set of properties.

\begin{framed}
\noindent
\textbf{Mechanism~2} (for chore division between two agents) \\

\noindent
\emph{Step~1}: Use Mechanism~1 to find an initial allocation of the chore, treating the chore valuations as cake valuations. \\

\noindent
\emph{Step~2}: Swap the pieces of the two agents in the allocation from Step~1.
\end{framed}

\begin{theorem}
\label{thm:twoagents-chore}
Mechanism~2 is a truthful, envy-free, and Pareto optimal chore division mechanism for two agents.
\end{theorem}

\begin{proof}
First, truthfulness holds because minimizing the chore in the swapped allocation is equivalent to maximizing the chore in the initial allocation, and Theorem~\ref{thm:twoagents-cake} shows that this is exactly what Mechanism~1 incentivizes the agents to do. Next, envy-freeness holds again by Theorem~\ref{thm:twoagents-cake} because getting at most half of the chore in the swapped allocation is equivalent to getting at least half of the chore in the initial allocation. Finally, in the initial allocation any interval of the chore that incurs a cost to only one agent is allocated to that agent, so in the swapped allocation the interval is allocated to the other agent, implying that the mechanism is Pareto optimal.
\end{proof}

Besides truthfulness, envy-freeness, and Pareto optimality, how do Mechanisms~1 and 2 fare with respect to the other properties defined in Section~\ref{sec:prelim}?
\begin{itemize}
\item Mechanism~1 is not \emph{anonymous}: If $W_1=[0,0.5]$ and $W_2=[0,1]$, both agents get value $0.5$, while if $W_1=[0,1]$ and $W_2=[0,0.5]$, $a_1$ gets value $0.75$ and $a_2$ gets value $0.25$.
\item It is also not \emph{position oblivious}: If $W_1=[0,0.5]$ and $W_2=[0,1]$, both agents get value $0.5$, while if $W_1=[0.5,1]$ and $W_2=[0,1]$, $a_1$ gets value $0.25$ and $a_2$ gets value $0.75$.
\item The allocation when $W_1=[0,1]$ and $W_2=[0,0.5]$ shows that the mechanism does not satisfy the \emph{connected piece assumption}.
\end{itemize}
The same examples demonstrate that Mechanism~2 likewise satisfies none of the three properties. As we show in the next section, these negative results are in fact not restricted to the two mechanisms that we consider here, but rather apply to all possible cake cutting and chore division mechanisms, even when envy-freeness is removed or significantly weakened.

\section{Impossibility Results}
\label{sec:impossibility}

In this section, we present a number of impossibility results on the existence of fair and truthful mechanisms that satisfy certain additional properties, for both cake cutting and chore division. In fact, one of the results holds even without envy-freeness, while for the remaining results we can replace envy-freeness by the much weaker notion of positive share. Interestingly, all of the impossibility results cease to hold if the mechanism is not required to allocate the entire resource, which again highlights the crucial difference that the free disposal assumption makes. We also observe that since these results are negative, they also hold for more general classes of valuations beyond piecewise uniform.

Before we introduce our results, some comments on truthful mechanisms are in order. At first glance, truthfulness may appear as a strong requirement that limits the feasible mechanisms to Mechanisms~1 and 2 and those returning a fixed allocation. However, the examples that we exhibit next show that the class of truthful mechanisms is surprisingly rich. We only describe cake cutting mechanisms, but it is possible to obtain corresponding chore division mechanisms via the reduction used in Mechanism~2. The latter two classes of mechanisms are adaptations of the picking and exchange mechanisms~\cite{AmanatidisBiCh17} to the divisible goods setting.
\begin{itemize}
\item A variant of Mechanism~1 where the two agents eat the cake at possibly different speeds fixed in advance. We can also partition the cake into a number of parts, ``flip'' some of these parts horizontally, ``paste'' some of the parts together, and run the variant of Mechanism~1 with different speeds for different parts.
\item A \emph{dictatorship} mechanism, where one fixed agent takes an arbitrary superset of her valued piece, and the other agent takes the rest of the cake.
\item A \emph{picking} mechanism, where the cake is cut in advance into a number of parts, some designated to $a_1$ and the rest to $a_2$. Each agent has a set of offers, each offer corresponding to a subset of the parts designated to her, and the agent picks an offer that she likes most. (For example, if parts 1, 3, 5, and 6 are designated to $a_1$, a possible set of offers is $\{\{1,3\},\{3,5\},\{1,5,6\}\}$. In this case, if $a_1$ picks the offer $\{3,5\}$, then $a_2$ gets parts 1 and 6.)
\item An \emph{exchange} mechanism, where the cake is cut in advance into a number of parts, some initially allocated to $a_1$ and the rest to $a_2$. Some exchange deals are considered; each exchange deal involves a subset of $a_1$'s parts and a subset of $a_2$'s parts, where each part appears in at most one exchange deal. An exchange deal materializes if the exchange benefits both agents.
\end{itemize}
It is also possible to combine these mechanisms by partitioning the cake in advance and running different mechanisms on different parts. Among these classes of mechanisms, only the first class, with the agents eating at the same speed, is envy-free.

We now proceed to our impossibility results, which demonstrate that even though there are several ways to obtain truthfulness, none of them is compatible with fairness and other desirable properties. We begin by showing that anonymity is directly at variance with truthfulness.

\begin{theorem}
\label{thm:anonymous}
There does not exist a truthful and anonymous cake cutting mechanism for two agents, even when each agent values a single interval of the form $[0,x_i]$.
\end{theorem}

\begin{proof}
Suppose that such a mechanism exists. Let $x\in(0,1)$ and $W_1=W_2=[0,x]$. Assume without loss of generality that in this instance, $a_1$ gets an interval containing point $x$ (perhaps as the left endpoint) and ending at point $x+f(x)>x$, possibly among other intervals. By anonymity, both agents must get half of the interval $[0,x]$.

If $W_1=[0,x+\epsilon]$ for some $\epsilon\in(0,f(x))$ and $W_2=[0,x]$, then $a_1$ must get value at least $x/2+\epsilon$, since otherwise she can manipulate by reporting $W_1=[0,x]$. Therefore $a_2$ gets value at most $x/2$ in this instance. By anonymity, if $W_1=[0,x]$ and $W_2=[0,x+\epsilon]$ for some $\epsilon\in(0,f(x))$, then $a_2$ gets value at least $x/2+\epsilon$ and $a_1$ at most $x/2$.

Now suppose that $W_1=W_2=[0,x+\epsilon]$ for some $\epsilon\in(0,f(x))$. By anonymity, both agents must get half of the interval $[0,x+\epsilon]$. If $a_1$ gets more than half of the interval $[x,x+\epsilon]$, then $a_2$ gets more than half of the interval $[0,x]$. In this case, if $W_2=[0,x]$, $a_2$ can manipulate by reporting $W_2=[0,x+\epsilon]$. So $a_1$ cannot get more than half of the interval $[x,x+\epsilon]$. By symmetry, neither can $a_2$. This means that both agents get exactly half of the interval $[x,x+\epsilon]$. In other words, for any $y\in(x,x+f(x))$, if $W_1=W_2=[0,y]$, then both agents receive exactly half of the interval $[x,y]$.

Next, consider the set
\[A:=\{(x,y)\in\mathbb{R}_{(0,1)}\times\mathbb{Q}_{(0,1)}\mid x<y<x+f(x)\}.\] This set is uncountable, since for each of the uncountably many $x$'s, there is at least one $y$ such that $(x,y)\in A$. If for each $y$ there only exist a finite number of $x$'s such that $(x,y)\in A$, this set would be countable, which we know is not the case. Hence there exists some $y$ such that  $(x,y)\in A$ for infinitely many $x$'s. Fix such a $y$.

Finally, suppose that $W_1=W_2=[0,y]$. For any of the infinitely many $x$'s such that $(x,y)\in A$, both agents must receive exactly half of the interval $[x,y]$. However, if the mechanism divides the interval $[0,y]$ into $k$ intervals in the allocation, then there can be at most one value of $x$ per interval, and therefore at most $k$ values in total, with this property. Since $k$ is finite, this gives us the desired contradiction.
\end{proof}

We remark that \emph{with} the free disposal assumption, Chen et al.'s mechanism is a truthful, envy-free, and anonymous cake cutting mechanism for two agents with arbitrary piecewise uniform valuations. The same authors showed that a particular extension of their mechanism, which allocates the desired pieces of the cake in the same way as their mechanism and allocates the undesired pieces of the cake in a certain simple way, is not truthful \cite[p.~296]{CLPP13}. Since any mechanism that allocates the desired pieces of the cake in this way is also anonymous, Theorem~\ref{thm:anonymous} shows that \emph{no} extension of Chen et al.'s mechanism can be truthful.

Next, we turn to the connected piece assumption and show that it is incompatible with truthfulness and positive share.

\begin{theorem}
\label{thm:connected}
There does not exist a truthful cake cutting mechanism for two agents that satisfies positive share and the connected piece assumption, even when each agent values a single interval of the form $[0,x_i]$.
\end{theorem}

\begin{proof}
Suppose that such a mechanism exists. First, we fix $x\in(0,1)$ and assume that $W_1=[0,x]$. For each $y\in(0,1)$, if $W_2=[0,y]$, then by the connected piece assumption, $a_2$ gets either a ``left interval'' of the form $[0,z]$ or a ``right interval'' of the form $[z,1]$; by positive share, we have $z\in(0,1)$, so $a_2$ never gets an interval that is simultaneously a left interval and a right interval. We claim that there exist $f(x),g(x)>0$ with $f(x)+g(x)<x$ such that if $W_2=[0,y]$ for some $0<y<f(x)$, then $a_2$ gets a superset of $[0,y]$, while if $W_2=[0,y]$ for some $f(x)\leq y < f(x)+g(x)$, then $a_2$ gets exactly $[0,f(x)]$. To prove this claim, we consider two cases.
\begin{itemize}
\item \emph{Case 1}: For some $y\in(0,1]$, if $W_2=[0,y]$, then $a_2$ gets a right interval. We show that this right interval must be the same for all such $y$. Suppose for contradiction that $a_2$ gets $[z_1,1]$ when $W_2=[0,y_1]$ and gets $[z_2,1]$ when $W_2=[0,y_2]$, for some $z_1<z_2$. By positive share, we have $z_2<y_2$. When $W_2=[0,y_2]$, $a_2$ can manipulate by reporting $W_2=[0,y_1]$ and obtain a higher utility. So we must have $z_1=z_2$. Hence there exists $h(x)\in(0,1)$ such that whenever $a_2$ gets a right interval, she gets exactly the interval $[h(x),1]$. In particular, if $W_2=[0,y]$ for some $y\in(0,h(x))$, then positive share implies that $a_2$ must get a left interval.

Next, define
\[f(x) := \sup_{y\in(0,1]}\{z\mid a_2\text{ gets the interval } [0,z] \text{ when } W_2 = [0,y]\}.\]
From the previous paragraph, we have $f(x)>0$; moreover, positive share implies that $f(x)\leq x$. If $W_2=[0,y]$ for some $y\in(0,f(x)]$, then $a_2$ must get the entire interval $[0,y]$ (and possibly more); otherwise, by definition of $f(x)$, she can manipulate to get the more of the interval $[0,y]$. In particular, when $W_2=[0,f(x)]$, the definition of $f(x)$ implies that $a_2$ gets exactly the interval $[0,f(x)]$. By positive share, this also means that $f(x)<x$.

Now, let $g(x)\in(0,h(x))$ be such that $f(x)+g(x)<x$. Suppose that $W_2=[0,y]$ for some $f(x)<y<f(x)+g(x)$. If $a_2$ gets a right interval, this interval must be $[h(x),1]$, which yields value strictly less than $f(x)$ to her. However, she can manipulate by reporting $W_2=[0,f(x)]$ and get value $f(x)$. So $a_2$ must get a left interval of length at least $f(x)$. By definition of $f(x)$, this must be exactly the interval $[0,f(x)]$. Therefore $f(x)$ and $g(x)$ satisfy the conditions of our claim.

\item \emph{Case 2}: For any $y\in(0,1]$, if $W_2=[0,y]$, then $a_2$ gets a left interval. We define $f(x)$ as in the previous case and choose any $g(x)>0$ such that $f(x)+g(x)<x$. A similar argument as before shows that $f(x)$ and $g(x)$ satisfy the conditions of our claim.
\end{itemize}

Next, let $x'\in(f(x),x)$. Since our choice of $x$ was arbitrary, for $x'$ there also exist $f(x')$ and $g(x')$ satisfying analogous conditions. Suppose first that $f(x')<f(x)$. Let $y\in(f(x'),f(x')+g(x'))$ be such that $y<f(x)$, and assume that $W_2=[0,y]$. If $W_1=[0,x]$, then since $y<f(x)$, our claim implies that $a_2$ gets at least the interval $[0,y]$. However, $a_1$ can manipulate by reporting $W_1=[0,x']$ so that $a_2$ only gets the interval $[0,f(x')]$. So we cannot have $f(x')<f(x)$. Likewise, if $f(x')>f(x)$, we can take $y\in (f(x),f(x)+g(x))$ to arrive at a contradiction. Hence we must have $f(x')=f(x)$. Since our choice of $x'$ was arbitrary, this holds for any $x'\in (f(x),x)$.

Finally, let $x''\in(0,f(x))$. Assume that when $W_1=[0,x'']$ and $W_2=[0,f(x)]$, $a_1$ gets value $r>0$. Let $r'\in(0,\min\{r,x-f(x)\})$. If $W_1=[0,f(x)+r']$ and $W_2=[0,f(x)]$, then since $f(f(x)+r')=f(x)$, $a_2$ gets exactly the interval $[0,f(x)]$, which leaves value at most $r'$ to $a_1$. However, in this case $a_1$ can manipulate by reporting $W_1=[0,x'']$ to obtain value $r>r'$. This is a contradiction.
\end{proof}

Bei et al. \cite{BCHTW17} showed that a similar impossibility result with envy-freeness instead of positive share holds \emph{with} the free disposal assumption, but using the larger class of piecewise constant valuations. For the class of valuations that we consider in Theorem~\ref{thm:connected}, there exists a simple truthful and envy-free mechanism that always returns a connected allocation assuming free disposal. The mechanism works as follows: Assume that agent $a_i$ declares $W_i=[0,x_i]$ for $i=1,2$. If $x_1\geq x_2$, allocate the interval $[x_1/2,x_1]$ to $a_1$ and $[0,x_1/2]$ to $a_2$; otherwise allocate the interval $[0,x_2/2]$ to $a_1$ and $[x_2/2,x_2]$ to $a_2$. One can check that this mechanism satisfies the claimed properties.

We now consider position obliviousness and show the non-existence of a truthful and position oblivious cake cutting mechanism that satisfies positive share for two agents. In Appendix~\ref{app:impossibility}, we prove a statement that uses proportionality instead of positive share but holds for any even number of agents.

\begin{theorem}
\label{thm:position-oblivious}
There does not exist a truthful and position oblivious cake cutting mechanism for two agents that satisfies positive share.
\end{theorem}

\begin{proof}
Suppose that such a mechanism exists. First, we claim that for any $x\in(0,1/3)$, if $W_1=W_2=[0,x]$, then one of the agents gets length strictly less than $x$ from the interval $[x,1]$. Assume that this is not the case, and suppose that $a_1$ receives value $r>0$ in the instance where $W_1=W_2=[0,x]$. Since both agents get length at least $x$ from $[x,1]$, we may assume without loss of generality that $a_1$ gets $[x,2x]$ (perhaps among other pieces). This means that if $W_1=[0,2x]$ and $W_2=[0,x]$, $a_1$ must receive value at least $x+r$; otherwise she can manipulate by reporting $W_1=[0,x]$. We now show that $a_1$ must still receive value at least $x+r$ when $W_1=W_2=[0,2x]$. This suffices to establish our claim since by symmetry, $a_2$ also receives value more than $x$ when $W_1=W_2=[0,2x]$, which is impossible.

If it were not the case that $a_1$ receives value at least $x+r$ when $W_1=W_2=[0,2x]$, then $a_2$ receives value more than $x-r$ in this instance. In the instance where $W_1=[0,2x]$, and $W_2$ has length $x$ and contains a piece of length more than $x-r$ that $a_2$ values in the previous instance, $a_2$ must receive value more than $x-r$; otherwise she can manipulate by reporting $W_2=[0,2x]$. Hence $a_1$ receives value less than $x+r$ in this instance. However, we know that $a_1$ receives value at least $x+r$ when $W_1=[0,2x]$ and $W_2=[0,x]$. This is a contradiction with position obliviousness. So $a_1$ indeed receives value at least $x+r$ when $W_1=W_2=[0,2x]$, proving our claim.

Next, fix $x\in(0,1/3)$. Assume without loss of generality that when $W_1=W_2=[0,x]$, $a_2$ gets length strictly less than $x$ from the interval $[x,1]$, and that $a_2$ does not get any piece outside the interval $[1-x,1]$. Consider the instance where $W_1=[0,1-x]$ and $W_2=[0,x]$. Note that $a_1$ must get value more than $1-2x$; otherwise she can manipulate by reporting $W_1=[0,x]$. Now, consider the instance where $W_1=W_2=[0,1-x]$. If $a_2$ gets value at least $x$, let $B_2$ be a piece of length $x$ that $a_2$ values and gets. If $W_1=[0,1-x]$ and $W_2=B_2$, $a_2$ must get at least the entire $B_2$, and hence $a_1$ gets value at most $1-2x$. However, we know that $a_1$ receives value greater than $1-2x$ when $W_1=[0,1-x]$ and $W_2=[0,x]$, which contradicts position obliviousness. Hence $a_2$ gets value less than $x$ when $W_1=W_2=[0,1-x]$. This means that when $W_1=[0,1-x]$ and $W_2=[0,1]$, $a_2$ gets value less than $2x$; otherwise $a_2$ can manipulate by reporting $W_2=[0,1]$ when in fact $W_2=[0,1-x]$.

Finally, assume that $a_2$ receives value $s>0$ when $W_1=W_2=[0,1]$. Let $s':=s/3 \in(0,1/3)$. If $a_1$ manipulates by reporting $W_1=[0,1-s']$, $a_2$ gets value less than $2s'<s$. Hence this is a beneficial manipulation for $a_1$, a contradiction.
\end{proof}

As with the connected piece assumption, Bei et al. \cite{BCHTW17} showed a similar negative result for position obliviousness and envy-freeness \emph{with} the free disposal assumption but using the larger class of piecewise constant valuations. For piecewise uniform valuations, Chen et al.'s mechanism is truthful, envy-free, and position oblivious under the free disposal assumption.

We end this section by showing that our impossibility results also carry over to chore division. The idea is the same as the one used in Mechanism~2, except that here we use it to establish negative results.

\begin{corollary}
\label{cor:chore-impossibility}
There does not exist a truthful chore division mechanism for two agents that satisfies any one of the following set of conditions: (i) anonymity; (ii) connected piece assumption and positive share; (iii) position obliviousness and positive share.
\end{corollary}

\begin{proof}
If there were a truthful  mechanism that satisfies one of the additional set of properties, we could obtain a cake cutting mechanism with the same properties as follows: First, we use the chore division mechanism to compute an initial allocation of the cake, treating the cake valuations as chore valuations. Then we swap the pieces of the two agents in this allocation. However, the existence of a cake cutting mechanism with these properties would contradict one of Theorems~\ref{thm:anonymous}, \ref{thm:connected}, and \ref{thm:position-oblivious}, respectively.
\end{proof}

\section{Bounds on the Efficiency Ratio}
\label{sec:competitive}

In this section, we investigate the efficiency of mechanisms satisfying various properties using the efficiency ratio, which is a standard measure for the welfare performance of mechanisms. In particular, we show that for both cake cutting and chore division, no truthful and envy-free mechanism attains a better efficiency ratio than the mechanisms we presented in Section~\ref{sec:mechanism-two}.

In contrast to the rest of the paper, in this section we do engage in interpersonal comparison of utility. For this reason, we assume that each agent has a value (or cost) of~1 for the entire resource, and scale the agent's utility for a subset of the resource proportionally.\footnote{For this to make sense, we also assume that each agent has a positive value (or cost) on a nonempty subset of the interval $[0,1]$.} In cake cutting, the \emph{social welfare} of an allocation is the sum of the values that the two agents receive, and the \emph{efficiency ratio} of a mechanism is the infimum, over all feasible instances, of the ratio between the social welfare of the allocation produced by the mechanism and the maximum social welfare. Similarly, in chore division, the \emph{social cost} of an allocation is the sum of the costs that the two agents incur, and the \emph{efficiency ratio} of a mechanism is the infimum, over all instances, of the ratio between the minimum social cost and the social cost of the allocation produced by the mechanism.\footnote{Here we assume that $\frac{0}{0}=1$. This is a reasonable choice because if the mechanism produces an allocation with zero social cost for a certain instance, there is no loss of efficiency in that instance.} Note that by definition, the efficiency ratio is always at most 1.

\subsection{Cake Cutting}

We start with cake cutting. Our first result gives a lower bound on the efficiency ratio of envy-free and Pareto optimal mechanisms.

\begin{theorem}
\label{thm:cake-EF-PO}
Any envy-free and Pareto optimal cake cutting mechanism for two agents has efficiency ratio at least $3/4$, and this bound is tight.
\end{theorem}

\begin{proof}
Fix an envy-free and Pareto optimal cake cutting mechanism. Consider an instance where the two agents value pieces of length $x\leq y$ respectively, and these pieces have an overlap of length $z\leq x$. The welfare-maximizing allocation gives the entire overlap to $a_1$, yielding welfare $1+\frac{y-z}{y}=2-\frac{z}{y}$. Next, we derive a lower bound for the welfare of the allocation that the mechanism produces. We consider two cases.

\begin{itemize}
\item \emph{Case 1:} $z\leq x/2$. Since the mechanism is Pareto optimal, the welfare achieved is at least the welfare when the entire overlap is allocated to $a_2$. This yields welfare $1+\frac{x-z}{x} = 2-\frac{z}{x}\geq\frac{3}{2}$. Hence the efficiency ratio is at least $\frac{3/2}{2}=\frac{3}{4}$.
\item \emph{Case 2:} $z>x/2$. The minimum welfare is achieved by minimizing $a_1$'s share of the overlap. Since the mechanism is envy-free, $a_1$ gets a piece of length at least $x/2$. This means that she gets length at least $\frac{x}{2}-(x-z) = z-\frac{x}{2}$ from the overlap. If she gets length exactly $z-\frac{x}{2}$ from the overlap, the social welfare is
\[\frac{1}{2}+\frac{y-z+x/2}{y}\geq \frac{3}{2}-\frac{z}{2y},\]
where the inequality holds since $x\geq z$. Hence the efficiency ratio is at least
\[\frac{\frac{3}{2}-\frac{z}{2y}}{2-\frac{z}{y}} = \frac{1}{2}\cdot\frac{3-\frac{z}{y}}{2-\frac{z}{y}}\geq \frac{3}{4}.\]
\end{itemize}
The tightness of the bound follows from the fact that Mechanism~1 has efficiency ratio $3/4$, which we show in Corollary~\ref{cor:cake-eating}.
\end{proof}

Next, we provide an upper bound on the efficiency ratio of truthful and envy-free mechanisms.

\begin{theorem}
\label{thm:cake-TR-EF}
Any truthful and envy-free cake cutting mechanism for two agents has efficiency ratio at most $3/4$, and this bound is tight.
\end{theorem}

\begin{proof}
Fix a truthful and envy-free cake cutting mechanism. Consider the instance where $W_1=W_2=[0,\epsilon]$ for some small $\epsilon>0$. Since the mechanism is envy-free, it must allocate exactly half of the interval $[0,\epsilon]$ to each agent. Assume without loss of generality that the mechanism gives $a_1$ at least half of the interval $[\epsilon,1]$; let $B_1\subseteq [\epsilon,1]$ be the piece that $a_1$ receives. Next, consider the instance where $W_1=[0,\epsilon]\cup B_1$ and $W_2=[0,\epsilon]$. By truthfulness, $a_1$ must receive $B_1$ and half of the interval $[0,\epsilon]$, and therefore $a_2$ can only receive half of $[0,\epsilon]$. For small $\epsilon$, the optimal social welfare approaches 2, while the social welfare of the allocation produced by the mechanism is at most $1/2+1 = 3/2$. Hence the efficiency ratio of the mechanism is at most $3/4$, as claimed.

The tightness of the bound follows from the fact that Mechanism~1 has efficiency ratio $3/4$, which we show in Corollary~\ref{cor:cake-eating}.
\end{proof}


The bounds in Theorems~\ref{thm:cake-EF-PO} and \ref{thm:cake-TR-EF} yield the following corollary, which immediately implies the tightness of both bounds. Moreover, it follows from the corollary that no truthful and envy-free cake cutting mechanism has a better efficiency ratio than Mechanism~1.

\begin{corollary}
\label{cor:cake-eating}
Any truthful, envy-free, and Pareto optimal cake cutting mechanism for two agents has efficiency ratio exactly $3/4$. In particular, the efficiency ratio of Mechanism~1 is $3/4$.
\end{corollary}

Next, we prove tight upper and lower bounds for mechanisms satisfying each of envy-freeness and Pareto optimality.

\begin{theorem}
\label{thm:cake-EF}
Any envy-free cake cutting mechanism for two agents has efficiency ratio at least $1/2$ and at most $\frac{2+\sqrt{3}}{4}\approx 0.933$, and both bounds are tight.
\end{theorem}

\begin{proof}
We consider the lower and upper bounds in turn.

\emph{Lower bound:} An envy-free mechanism always gives each agent a utility of at least $1/2$, so the social welfare of the resulting allocation is at least 1. On the other hand, the optimal social welfare is at most 2. Hence the efficiency ratio is at least $1/2$. To show that this bound is tight, consider a mechanism that always outputs the allocation output by the Mechanism~1, except when $W_1=[0,1]$ and $W_2=[0,\epsilon]$ for some small $\epsilon>0$. For this instance, the mechanism gives the piece $[\epsilon/2,(1+\epsilon)/2]$ to $a_2$ and the rest of the cake to $a_1$. This is an envy-free allocation, and both agents receive utility $1/2$. On the other hand, for small $\epsilon$, the optimal welfare in this instance approaches 2.

\emph{Upper bound:} Caragiannis et al.~\cite[Theorem~9]{CaragiannisKaKa11} showed that for two agents with arbitrary (not necessarily piecewise uniform) valuations, the maximum welfare of an envy-free allocation is no less than $\frac{2+\sqrt{3}}{4}$ times the optimal social welfare. Hence a mechanism that always returns an envy-free allocation with maximum welfare has efficiency ratio at least $\frac{2+\sqrt{3}}{4}$. To show that this bound is tight, we adapt the example of Caragiannis et al. to piecewise uniform valuations. Consider the instance where $W_1=[0,\sqrt{3}-1]$ and $W_2=[0,1]$. The welfare-maximizing allocation gives the entire overlap to $a_1$, yielding welfare $1+(2-\sqrt{3})=3-\sqrt{3}$. On the other hand, any envy-free allocation gives $a_1$ a piece of length at most $1/2$ from the overlap. So the maximum welfare of an envy-free allocation is $\frac{1/2}{\sqrt{3}-1}+\frac{1}{2}=\frac{3+\sqrt{3}}{4}$. Hence the efficiency ratio of any envy-free mechanism is at most $\frac{(3+\sqrt{3})/4}{3-\sqrt{3}}=\frac{2+\sqrt{3}}{4}$, as claimed.
\end{proof}

\begin{theorem}
\label{thm:cake-PO}
Any Pareto optimal cake cutting mechanism for two agents has efficiency ratio at least $1/2$ and at most $1$, and both bounds are tight.
\end{theorem}

\begin{proof}
We consider the lower and upper bounds in turn.

\emph{Lower bound:} Consider an instance where the two agents value pieces of length $x\leq y$ respectively, and these pieces have an overlap of length $z\leq x$. The welfare-maximizing allocation gives the entire overlap to $a_1$, yielding welfare $1+\frac{y-z}{y}=2-\frac{z}{y}$. On the other hand, the lowest welfare of a Pareto optimal allocation is achieved by giving the entire overlap to $a_2$, which yields welfare $1+\frac{x-z}{x}=2-\frac{z}{x}$. Hence the efficiency ratio of a Pareto optimal mechanism is at least $\frac{2-z/x}{2-z/y}\geq\frac{2-1}{2-0}=\frac{1}{2}$. To show that this bound is tight, consider a mechanism that always outputs the allocation output by Mechanism~1, except when $W_1=[0,1]$ and $W_2=[0,\epsilon]$ for some small $\epsilon>0$. For this instance, the mechanism gives the entire cake to $a_1$. This is a Pareto optimal allocation with social welfare 1. On the other hand, for small $\epsilon$, the optimal welfare approaches 2.

\emph{Upper bound:} By definition, the efficiency ratio of any mechanism cannot exceed 1. Moreover, a mechanism that always returns an allocation with maximum social welfare is Pareto optimal and has efficiency ratio exactly 1.
\end{proof}

\subsection{Chore Division}

We now address chore division and exhibit a surprising contrast to cake cutting. Our first result shows that unlike for cake cutting, for chore division no truthful and envy-free mechanism has a positive efficiency ratio. Since Mechanism~2 is truthful and envy-free, the result applies to the mechanism as well.

\begin{theorem}
\label{thm:chore-TR-EF}
Any truthful and envy-free chore division mechanism for two agents has efficiency ratio 0. In particular, the efficiency ratio of Mechanism~2 is 0.
\end{theorem}

\begin{proof}
Fix a truthful and envy-free chore division mechanism. Consider the instance where $W_1=W_2=[0,\epsilon]$ for some small $\epsilon>0$. Since the mechanism is envy-free, it must allocate exactly half of the interval $[0,\epsilon]$ to each agent. Assume without loss of generality that the mechanism gives $a_2$ at least half of the interval $[\epsilon,1]$; let $B_2\subseteq [\epsilon,1]$ be the piece that $a_2$ receives. Next, consider the instance where $W_1=[0,\epsilon]\cup B_2$ and $W_2=[0,\epsilon]$. By truthfulness, $a_1$ must receive at most half of the interval $[0,\epsilon]$, and therefore $a_2$ receives at least half of the interval $[0,\epsilon]$. For small $\epsilon$, the minimum social cost in this instance approaches 0, while the social cost of the allocation produced by the mechanism is at least $1/2$. Hence the efficiency ratio of the mechanism is 0, as claimed.
\end{proof}

Theorem~\ref{thm:chore-TR-EF} also implies that any combination of truthfulness, envy-freeness, and Pareto optimality does not suffice to guarantee a positive efficiency ratio. On the other hand, a mechanism that always returns an allocation with minimum social cost is Pareto optimal and has the highest possible efficiency ratio of 1. Our next result shows that an envy-free mechanism can achieve an efficiency ratio of up to $8/9$.

\begin{theorem}
\label{thm:chore-EF}
Any envy-free chore division mechanism for two agents has efficiency ratio at most $8/9$, and this bound is tight.
\end{theorem}

\begin{proof}
Caragiannis et al.~\cite[Theorem~17]{CaragiannisKaKa11} showed that for two agents with arbitrary (not necessarily piecewise uniform) valuations, the minimum cost of an envy-free allocation is no more than $9/8$ times the optimal social cost. Hence a mechanism that always returns an envy-free allocation with minimum cost has efficiency ratio $8/9$. To show that this bound is tight, we adapt the example of Caragiannis et al. to piecewise uniform valuations. Consider the instance where $W_1=[0,2/3]$ and $W_2=[0,1]$. The cost-minimizing allocation gives the entire overlap to $a_2$, yielding cost $2/3$. On the other hand, any envy-free allocation gives $a_2$ a piece of length at most $1/2$ from the overlap. So the minimum cost of an envy-free allocation is $\frac{1/6}{2/3}+\frac{1}{2}=\frac{3}{4}$. Hence the efficiency ratio of any envy-free mechanism is at most $\frac{2/3}{3/4}=\frac{8}{9}$, as claimed.
\end{proof}

\section{Extensions to Multiple Agents}
\label{sec:mechanism-many}

In this section, we consider the general setting where we allocate the resource among any number of agents. We assume that each agent $a_i$ only values (or has cost on) the interval $[0,x_i]$ for some $x_i$. Such valuations may appear in a scenario where the agents are dividing machine processing time: agent $a_i$ has a deadline $x_i$ for her jobs, so she would like to maximize the processing time she gets before $x_i$ but has no value for any processing time after $x_i$.
We also remark that the example used to illustrate that removing the free disposal assumption can be problematic consists of \emph{two} agents whose valuations belong to this class \cite[p.~296]{CLPP13}. Hence, designing a fair and truthful algorithm is by no means an easy problem even for this valuation class.

We first describe the cake cutting mechanism.

\begin{framed}
\noindent
\textbf{Mechanism~3} (for cake cutting among $n$ agents) \\

\noindent
\emph{Step~1}: If there is one agent left, the agent gets the entire remaining cake. Else, assume that there are $k\geq 2$ agents and length $\ell$ of the cake left. Find the maximum $x\in[0,\ell]$ such that agent $i$ values the entire interval $[(i-1)x,ix]$ for all $i=1,2,\dots,k$, and allocate the interval $[(i-1)x,ix]$ to agent $i$. \\

\noindent
\emph{Step~2}: The agent whose right endpoint of her allocated interval coincides with the right endpoint of her valued piece exits the process. If there are more than one such agent, choose the one with the lowest number.\footnotemark \\

\noindent
\emph{Step~3}: Renumber the remaining agents in the same order starting from $1$, and relabel the left endpoint of the remaining cake as point $0$. Return to Step~1.
\end{framed}
\footnotetext{There always exists at least one such agent, since otherwise the value of $x$ in Step~1 can be increased.}

\begin{theorem}
\label{thm:many-cake}
Let $n$ be any positive integer. Mechanism~3 is a truthful, envy-free, and Pareto optimal cake cutting mechanism for $n$ agents, if each agent only values a single interval of the form $[0,x_i]$.
\end{theorem}

\begin{proof}
First, for truthfulness, there are two types of manipulation: moving $x_i$ to the left and to the right. Moving $x_i$ to the left can only cause $a_i$ to quit the process early when she could have gained more by staying on. On the other hand, if moving $x_i$ to the right causes the allocation to change in some round of Step 1, the agent can only get less value from the allocated interval as its right endpoint moves past $x_i$. Moreover, since she has no more valued intervals to the right, she cannot make up for the loss.

Next, for envy-freeness, if an agent is no longer in the process, she has no more piece of value. During the process, in each round all remaining agents receive an interval of the same length. Since each agent values the entire interval that she receives, she does not envy any other agent.

Finally, our mechanism allocates any interval valued by at least one agent to an agent who values it. This establishes Pareto optimality.
\end{proof}

We remark that constructing truthful and envy-free mechanisms that work beyond the class of valuations in Theorem~\ref{thm:many-cake} appear to be highly nontrivial.
For instance, even if every agent only values a single interval (not necessarily starting at $0$), then a mechanism that tries to find valued intervals of equal length according to the agent ordering no longer works: if $W_1=[0.5,1]$ and $W_2=[0,0.5]$, such a mechanism would not be able to allocate any of the cake.
Alternatively, one could try a generalization of Mechanism~3 that finds the smallest $x$ such that $ix$ is the right endpoint of some agent $i$'s valued interval. 
However, if $W_1=W_2=[0.5,1]$, then this mechanism would allocate $[0,0.5]$ to agent 1 and $[0.5,1]$ to agent 2, thereby leaving agent 1 envious.

Unlike in the case of two agents, there is no simple reduction between cake cutting and chore division in the general case. Nevertheless, our next result shows a truthful and proportional chore division mechanism for any number of agents.\footnote{We are grateful to Bo Li for pointing out an error in an earlier version of this mechanism.} We were not able to strengthen the proportionality guarantee to envy-freeness and leave it as an interesting open question for future research.

\begin{framed}
\noindent
\textbf{Mechanism~4} (for chore division among $n$ agents) \\

\noindent
\emph{Step~1}: Let $a_1$ take the piece $[0,x_1/n]\cup [x_1,1]$. If some other agent has no cost on parts of the interval $[0,x_1/n]$, give those parts to the agent. (If there are several such agents, allocate the parts arbitrarily.) \\

\noindent
\emph{Step~2}: Proceed similarly with the next agent up to $a_{n-1}$ and the remaining chore; agent $a_i$ takes the leftmost interval with value $y_i/n$ as well as any piece for which she has no cost, where $y_i:=\min(x_1,x_2,\dots,x_i)$. (If $a_i$ has cost less than $y_i/n$ left, she takes the entire remaining chore.) \\

\noindent
\emph{Step~3}: Agent $a_n$ takes all of the remaining chore.
\end{framed}

\begin{theorem}
\label{thm:many-chore}
Let $n$ be any positive integer. Mechanism~4 is a truthful, proportional, and Pareto optimal chore division mechanism for $n$ agents, if each agent only has cost on a single interval of the form $[0,x_i]$.
\end{theorem}

\begin{proof}
We begin with truthfulness. First, any agent who has no cost on some piece that the mechanism initially allocates to another agent has no incentive not to take the piece. Apart from this, agent $a_n$ has no control over her allocation, so the mechanism is truthful for her. For any other agent, there are two types of manipulation: moving $x_i$ to the left and to the right. Moving $x_i$ to the right can only increase the value of the piece that $a_i$ has to take. If $a_i$ moves $x_i$ to the left while staying to the right of $y_{i-1}=\min(x_1,x_2,\dots,x_{i-1})$, nothing changes. Else, she moves $x_i$ by an amount $z$ on the left of $y_{i-1}$. In this case, she can save a cost of at most $z/n$ but has to take a piece of cost $z$ at the end. So $a_i$ does not have a profitable manipulation.

We now consider proportionality. Each agent $a_i$ up to $a_{n-1}$ gets a piece of cost at most $y_i/n\leq x_i/n$. For $a_n$, we consider two cases. If $x_n\leq y_{n-1}$, then each of the first $n-1$ agents takes at least $1/n$ of the interval $[0,x_n]$, so at most $1/n$ of this interval is left for $a_n$. Else, we have $x_n>y_{n-1}$. The intervals $[0,(n-1)y_{n-1}/n]$ and $[y_{n-1},1]$ will not be left to $a_n$, meaning that $a_n$ incurs a cost of at most $y_{n-1}/n<x_n/n$.

Finally, our mechanism allocates any interval for which some agent has no cost to one such agent. This establishes Pareto optimality.
\end{proof}

\section{Conclusion and Future Work}

In this paper, we study the problem of fairly dividing a heterogeneous resource in the presence of strategic agents and demonstrate the powers and limitations of truthful mechanisms in this setting. An immediate question that arises is whether the mechanisms in Section~\ref{sec:mechanism-two} can be generalized to work for any number of agents with piecewise uniform valuations. While our results in Section~\ref{sec:mechanism-many} provide a partial answer to this question, extending to the general setting seems to require a drastically different idea. Indeed, while other examples of truthful mechanisms given in Section~\ref{sec:impossibility} can be generalized to multiple agents, these extensions do not satisfy envy-freeness or even positive share. Of course, it could also be that there is an impossibility result once we move beyond the case of two agents.

Another direction for future work is to allow agents to have valuations from a larger class. A natural next step would be to consider the class of \emph{piecewise constant} valuations, in which an agent values each interval uniformly but can have different marginal utilities for different intervals. Intriguingly, it is not known whether there exists a deterministic truthful and envy-free mechanism even for two agents with piecewise constant valuations, either with or without the free disposal assumption. Again, the question is still open even if we relax envy-freeness to positive share; it does not seem clear whether (significantly) relaxing the fairness notion helps in designing a truthful mechanism for piecewise constant valuations.

Finally, our results in Section~\ref{sec:competitive} leave open the question of determining the best efficiency ratio that can be achieved by a truthful mechanism. This quantity has been called the \emph{price of truthfulness} by Maya and Nisan~\cite{MN12}, who studied it in the setting where free disposal is allowed. While we have established tight upper bounds on the efficiency ratio of truthful and envy-free mechanisms for both cake cutting and chore division, our current techniques are not sufficient to show that the bounds remain tight in the absence of envy-freeness. If a truthful mechanism with a higher efficiency ratio were to exist, it would be useful in situations where only resistance to strategic behavior is desired.

\section*{Acknowledgments}

This work was partially supported by the European Research Council (ERC) under grant number 639945 (ACCORD) and by a Stanford Graduate Fellowship. A preliminary version of the paper appeared in Proceedings of the 27th International Joint Conference on Artificial Intelligence. The authors thank Felix Brandt, Dominik Peters, and Christian Stricker for helpful discussions, and the anonymous reviewers for valuable comments.

\bibliographystyle{abbrv}
\bibliography{refs-full}

\begin{thebibliography}{10}

\bibitem{AFGST17}
R.~Alijani, M.~Farhadi, M.~Ghodsi, M.~Seddighin, and A.~S. Tajik.
\newblock Envy-free mechanisms with minimum number of cuts.
\newblock In {\em Proceedings of the 31st AAAI Conference on Artificial
  Intelligence}, pages 312--318, 2017.

\bibitem{AmanatidisBiCh17}
G.~Amanatidis, G.~Birmpas, G.~Christodoulou, and E.~Markakis.
\newblock Truthful allocation mechanisms without payments: Characterization and
  implications on fairness.
\newblock In {\em Proceedings of the 18th ACM Conference on Economics and
  Computation}, pages 545--562, 2017.

\bibitem{AmanatidisBiMa16}
G.~Amanatidis, G.~Birmpas, and E.~Markakis.
\newblock On truthful mechanisms for maximin share allocations.
\newblock In {\em Proceedings of the 25th International Joint Conference on
  Artificial Intelligence}, pages 31--37, 2016.

\bibitem{AzizBoMo17}
H.~Aziz, A.~Bogomolnaia, and H.~Moulin.
\newblock Fair mixing: the case of dichotomous preferences.
\newblock In {\em Proceedings of the 20th ACM Conference on Economics and
  Computation}, pages 753--781, 2019.

\bibitem{aziz2016discrete}
H.~Aziz and S.~Mackenzie.
\newblock A discrete and bounded envy-free cake cutting protocol for any number
  of agents.
\newblock In {\em Proceedings of the 57th Annual Symposium on Foundations of
  Computer Science}, pages 416--427, 2016.

\bibitem{aziz2015discrete}
H.~Aziz and S.~Mackenzie.
\newblock A discrete and bounded envy-free cake cutting protocol for four
  agents.
\newblock In {\em Proceedings of the 48th Annual ACM SIGACT Symposium on Theory
  of Computing}, pages 454--464, 2016.

\bibitem{aziz2014cake}
H.~Aziz and C.~Ye.
\newblock Cake cutting algorithms for piecewise constant and piecewise uniform
  valuations.
\newblock In {\em Proceedings of the 10th International Conference on Web and
  Internet Economics}, pages 1--14, 2014.

\bibitem{Bade15}
S.~Bade.
\newblock Multilateral matching under dichotomous preferences.
\newblock Working paper, 2015.

\bibitem{BCHTW17}
X.~Bei, N.~Chen, G.~Huzhang, B.~Tao, and J.~Wu.
\newblock Cake cutting: Envy and truth.
\newblock In {\em Proceedings of the 26th International Joint Conference on
  Artificial Intelligence}, pages 3625--3631, 2017.

\bibitem{BeiLuMa19}
X.~Bei, X.~Lu, P.~Manurangsi, and W.~Suksompong.
\newblock The price of fairness for indivisible goods.
\newblock In {\em Proceedings of the 28th International Joint Conference on
  Artificial Intelligence}, pages 81--87, 2019.

\bibitem{BogomolnaiaMo01}
A.~Bogomolnaia and H.~Moulin.
\newblock A new solution to the random assignment problem.
\newblock {\em Journal of Economic Theory}, 100(2):295--328, 2001.

\bibitem{BogomolnaiaMo04}
A.~Bogomolnaia and H.~Moulin.
\newblock Random matching under dichotomous preferences.
\newblock {\em Econometrica}, 72(1):257--279, 2004.

\bibitem{BogomolnaiaMoSt05}
A.~Bogomolnaia, H.~Moulin, and R.~Stong.
\newblock Collective choice under dichotomous preferences.
\newblock {\em Journal of Economic Theory}, 122(2):165--184, 2005.

\bibitem{BT95}
S.~J. Brams and A.~D. Taylor.
\newblock An envy-free cake division protocol.
\newblock {\em American Mathematical Monthly}, 102(1):9--18, 1995.

\bibitem{brams1996fair}
S.~J. Brams and A.~D. Taylor.
\newblock {\em Fair Division: From Cake-Cutting to Dispute Resolution}.
\newblock Cambridge University Press, 1996.

\bibitem{branzei2015dictatorship}
S.~Br{\^a}nzei and P.~B. Miltersen.
\newblock A dictatorship theorem for cake cutting.
\newblock In {\em Proceedings of the 24th International Joint Conference on
  Artificial Intelligence}, pages 482--488, 2015.

\bibitem{CaragiannisKaKa11}
I.~Caragiannis, C.~Kaklamanis, P.~Kanellopoulos, and M.~Kyropoulou.
\newblock The efficiency of fair division.
\newblock {\em Theory of Computing Systems}, 50(4):589--610, 2011.

\bibitem{CLPP13}
Y.~Chen, J.~K. Lai, D.~C. Parkes, and A.~D. Procaccia.
\newblock Truth, justice, and cake cutting.
\newblock {\em Games and Economic Behavior}, 77:284--297, 2013.

\bibitem{dehghani2018chore}
S.~Dehghani, A.~Farhadi, M.~Hajiaghayi, and H.~Yami.
\newblock Envy-free chore division for an arbitrary number of agents.
\newblock In {\em Proceedings of the 29th Annual ACM-SIAM Symposium on Discrete
  Algorithms}, pages 2564--2583, 2018.

\bibitem{dubins1961cut}
L.~E. Dubins and E.~H. Spanier.
\newblock How to cut a cake fairly.
\newblock {\em American Mathematical Monthly}, 68(1):1--17, 1961.

\bibitem{Duddy15}
C.~Duddy.
\newblock Fair sharing under dichotomous preferences.
\newblock {\em Mathematical Social Sciences}, 73:1--5, 2015.

\bibitem{farhadi2017complexity}
A.~Farhadi and M.~Hajiaghayi.
\newblock On the complexity of chore division.
\newblock In {\em Proceedings of the 27th International Joint Conference on
  Artificial Intelligence}, pages 226--232, 2018.

\bibitem{GoldbergHoSu20}
P.~W. Goldberg, A.~Hollender, and W.~Suksompong.
\newblock Contiguous cake cutting: Hardness results and approximation
  algorithms.
\newblock In {\em Proceedings of the 34th AAAI Conference on Artificial
  Intelligence}, 2020.
\newblock Forthcoming.

\bibitem{heydrich2015dividing}
S.~Heydrich and R.~van Stee.
\newblock Dividing connected chores fairly.
\newblock {\em Theoretical Computer Science}, 593:51--61, 2015.

\bibitem{kurokawa2013cut}
D.~Kurokawa, J.~K. Lai, and A.~D. Procaccia.
\newblock How to cut a cake before the party ends.
\newblock In {\em Proceedings of the 27th AAAI Conference on Artificial
  Intelligence}, pages 555--561, 2013.

\bibitem{MN12}
A.~Maya and N.~Nisan.
\newblock Incentive compatible two player cake cutting.
\newblock In {\em Proceedings of the 8th International Workshop on Internet and
  Network Economics}, pages 170--183, 2012.

\bibitem{Menon2017}
V.~Menon and K.~Larson.
\newblock Deterministic, strategyproof, and fair cake cutting.
\newblock In {\em Proceedings of the 26th International Joint Conference on
  Artificial Intelligence}, pages 352--358, 2017.

\bibitem{MT10}
E.~Mossel and O.~Tamuz.
\newblock Truthful fair division.
\newblock In {\em Proceedings of the 3rd International Symposium on Algorithmic
  Game Theory}, pages 288--299, 2010.

\bibitem{moulin2004fair}
H.~Moulin.
\newblock {\em Fair Division and Collective Welfare}.
\newblock MIT press, 2004.

\bibitem{peterson1998exact}
E.~Peterson and F.~E. Su.
\newblock Exact procedures for envy-free chore division.
\newblock Working paper, 1998.

\bibitem{peterson2002four}
E.~Peterson and F.~E. Su.
\newblock Four-person envy-free chore division.
\newblock {\em Mathematics Magazine}, 75(2):117--122, 2002.

\bibitem{Pro16}
A.~D. Procaccia.
\newblock Cake cutting algorithms.
\newblock In F.~Brandt, V.~Conitzer, U.~Endriss, J.~Lang, and A.~D. Procaccia,
  editors, {\em Handbook of Computational Social Choice}, chapter~13, pages
  311--329. Cambridge University Press, Cambridge, 2016.

\bibitem{RW98}
J.~Robertson and W.~Webb.
\newblock {\em Cake-Cutting Algorithms: Be Fair if You Can}.
\newblock Peters/CRC Press, 1998.

\bibitem{steinhaus48}
H.~Steinhaus.
\newblock The problem of fair division.
\newblock {\em Econometrica}, 16(1):101--104, 1948.

\bibitem{Stro80}
W.~Stromquist.
\newblock How to cut a cake fairly.
\newblock {\em American Mathematical Monthly}, 87(8):640--644, 1980.

\bibitem{su1999}
F.~E. Su.
\newblock Rental harmony: Sperner's lemma in fair division.
\newblock {\em American Mathematical Monthly}, 106(10):930--942, 1999.

\end{thebibliography}

\appendix

\section{Additional Impossibility Result}
\label{app:impossibility}

\begin{theorem}
\label{thm:position-oblivious-app}
Let $n=2k$ for some positive integer $k$. There does not exist a truthful, proportional, and position oblivious cake cutting mechanism for $n$ agents.
\end{theorem}

\begin{proof}
Suppose that such a mechanism exists. For simplicity, we represent the cake by the interval $[0,4k^2+k]$; this can be easily normalized back to $[0,1]$.

First, consider the instance where $W_{2i-1}=W_{2i}=[i-1,i]$ for $i=1,2,\dots,k$. Since the interval $[k,4k^2+k]$ is of length $4k^2$ and there are $2k$ agents, some agent gets value more than $2k-1$ from the interval. Assume without loss of generality that $a_1$ is one such agent, and that $a_1$ gets the interval $[k,3k-1]$. Since the mechanism is proportional, $a_1$ must get value at least $1/2k$ from the interval $[0,1]$ as well.

Next, consider the instance where $W_1=[0,1]\cup[k,3k-1]$, $W_2=[0,1]$, and $W_{2i-1}=W_{2i}=[i-1,i]$ for $i=2,3,\dots,k$. Agent $a_1$ must still get value at least $1/2k$ from the interval $[0,1]$; otherwise she can report $W_1=[0,1]$ instead. This means that $a_2$ gets a total value of at most $1-1/2k$ in this instance.

Finally, consider the instance where $W_1=W_2=[0,1]\cup[k,3k-1]$ and $W_{2i-1}=W_{2i}=[i-1,i]$ for $i=2,3,\dots,k$. By proportionality, $a_2$ must receive value at least $1$; let $B_2\subseteq [0,1]\cup[k,3k-1]$ be a piece of length $1$ that $a_2$ receives. If $W_2=B_2$ while the other $W_i$'s remain fixed, then since the mechanism is position oblivious, $a_2$ must get a total value of at most $1-1/2k$. However, in that case $a_2$ can report $W_2=[0,1]\cup[k,3k-1]$ and receive value $1$. This implies that the mechanism is not truthful and yields the desired contradiction.
\end{proof}

\end{document}